\documentclass[journal, onecolumn, 11pt]{IEEEtran}
\IEEEoverridecommandlockouts
\usepackage{cite}
\usepackage{tikz}\usetikzlibrary{matrix,decorations.pathreplacing,positioning}
\usepackage{amsmath,amsthm,amssymb,mathabx,bm,relsize,wasysym}
\usepackage{enumitem}

\usepackage{calrsfs}

\usepackage[margin=3cm]{geometry}

\definecolor{myg}{RGB}{220,220,220}
 
\theoremstyle{definition}
\newtheorem{theorem}{Theorem}[section]
\newtheorem{corollary}[theorem]{Corollary}
\newtheorem{proposition}[theorem]{Proposition}

\newtheorem{definition}[theorem]{Definition}
\newtheorem{example}[theorem]{Example}
\newtheorem{notation}[theorem]{Notation}

\newtheorem{claim}{Claim}

\newcommand*{\myproofname}{Proof of the claim}
\newenvironment{clproof}[1][\myproofname]{\begin{proof}[#1]}{\end{proof}}

\newcommand{\numberset}{\mathbb}

\newcommand{\R}{\numberset{R}}

\newcommand{\mS}{\mathcal{S}}

\newcommand{\mA}{\mathcal{A}}
\newcommand{\mN}{\mathcal{N}}

\newcommand{\mF}{\mathcal{F}}

\newcommand{\mD}{\mathcal{D}}
\newcommand{\mU}{\mathcal{U}}

\newcommand{\dH}{d_\textnormal{H}}

\newcommand{\mE}{\mathcal{E}}

\newcommand{\mV}{\mathcal{V}}

\newcommand{\inn}{\textnormal{in}}
\newcommand{\out}{\textnormal{out}}

\newcommand{\CA}{\textnormal{C}_1}

\newcommand{\adv}{\textnormal{\textbf{A}}}

\newlength{\mynodespace}
\newcommand{\ein}{\text{deg}^{-}}
\newcommand{\eout}{\text{deg}^{+}}

\begin{document}

\title{The Curious Case of the Diamond Network}

\author{Allison Beemer and Alberto Ravagnani
\thanks{A. Beemer is with the Department of Mathematics, University of Wisconsin-Eau Claire, U.S.A.; A. Ravagnani is with the Department of Mathematics and Computer Science, Eindhoven University of Technology, the Netherlands. A. Ravagnani is partially supported by the Dutch Research Council through grant OCENW.KLEIN.539.}
}

\maketitle

\begin{abstract}
This work considers the one-shot capacity of communication networks subject to adversarial noise affecting a subset of network edges. In particular, we examine previously-established upper bounds on one-shot capacity. We introduce the Diamond Network as a minimal example to show that known cut-set bounds are not sharp in general. We then give a capacity-achieving scheme for the Diamond Network that implements an adversary detection strategy. Finally, we give a sufficient condition for tightness of the Singleton Cut-Set Bound in a family of two-level networks.
\end{abstract}

\begin{IEEEkeywords}
Network coding, adversarial network, capacity, cut-set bound, Singleton bound
\end{IEEEkeywords}

\section{Introduction}
As the prevalence of interconnected devices grows, vulnerable communication networks must be able to counter the actions of malicious actors; a unified understanding of the fundamental communication limits of these networks is therefore paramount.
The correction of errors introduced by adversaries in networks has been studied in a number of previous works. Cai and Yeung give generalizations of several classical coding bounds to the network setting in \cite{YC06,CY06}. Refined bounds and related code constructions for adversarial networks are presented in, e.g., \cite{YY07, JLKHKM07, M07, YNY07, YYZ08, RK18}. 
The work most closely related to this paper is \cite{RK18}, where a unified combinatorial framework for adversarial networks and a method for porting point-to-point coding-theoretic results to the network setting are established. In contrast to works that address random errors in networks, or a combination of random and adversarial errors, \cite{RK18} focuses purely on adversarial, or \textit{worst-case}, errors.
%, and seeks bounds on single-use \textit{zero-error} capacity \cite{S56}. 
The results presented here assume the same model in a single-use regime.

We focus on networks whose inputs are drawn from a finite alphabet and whose intermediate nodes may process information before forwarding. We assume that an omniscient adversary can corrupt up to some fixed number of alphabet symbols sent along a subset of network edges. The \textit{one-shot capacity} of such an adversarial network measures the number of symbols that can be sent with zero error during a single transmission round. A universal approach to forming \textit{cut-set bounds}, which are derived by reducing the capacity problem to a minimization across cut-sets of the underlying directed graph of the network, is presented in \cite{RK18}. Any coding-theoretic bound may be ported to the networking setting, including the famous
\textit{Singleton Bound}.
%, but in this paper we focus on the cut-set bound derived from the well-known 
By exhibiting a minimal example, we show that even when the Singleton bound gives the best established upper bound on one-shot capacity for a network, it is not always tight (regardless of the size of the network alphabet). Our example, which we call the \textit{Diamond Network}, requires that a single symbol be sacrificed to the task of locating the adversary within the network. Interestingly, this requirement results in a non-integer-valued one-shot capacity. 

We note that the requirement that the receiver locate the adversary is related to the problem of \textit{authentication} in networks (see, e.g. \cite{KK16,SBDP19, BGKKY20}). In our capacity-achieving scheme for the Diamond Network, one intermediate vertex must be able to either sound an alarm (if the adversary is detected), or decode correctly (when the adversary is absent). On the other hand, in our presented scheme for a modification of the Diamond Network, called the \textit{Mirrored Diamond Network}, the way in which intermediate vertices sound the alarm must simultaneously serve as the way in which a particular alphabet symbol is transmitted. This interplay between authentication and correction is reminiscent of the work in \cite{BGKKY20}, where the idea of \textit{partial correction} over arbitrarily-varying multiple-access channels is introduced.

This paper is organized as follows. In Section \ref{sec:prelims} we introduce necessary notation and background. Sections \ref{sec:min-achiev} and \ref{sec:min-conv} together establish the exact one-shot capacity of the Diamond Network, proving that the Singleton Cut-Set Bound is not tight. In Section \ref{sec:mirrored-diamond}, we establish the (bound-achieving) one-shot capacity of the Mirrored Diamond Network. Section \ref{sec:2-level} expands our focus to the broader class of \textit{two-level} networks, and gives a sufficient condition for a network in this class to meet the best cut-set bound. We conclude and give future directions in Section \ref{sec:conclusion}.

\section{Preliminaries}%"Problem Statement"?
\label{sec:prelims}
% setup, notation
% cut-set bound
We introduce the terminology and notation for the remainder of the paper. We start by formally defining communication networks as in~\cite{RK18}.

\begin{definition}
A (\textbf{single-source communication)} \textbf{network} is a 4-tuple $\mN=(\mV,\mE,S, {\bf T})$, where:
\begin{itemize}
    \item[(A)] $(\mV,\mE)$ is a finite, directed and acyclic multigraph;
    \item[(B)] $S \in \mV$ is the \textbf{source};
    \item[(C)] ${\bf T} \subseteq \mV$ is the set of \textbf{terminals}.
\end{itemize}
We also assume the following:
\begin{itemize}
    \item[(D)] $|{\bf T}| \ge 1$ and $S \notin {\bf T}$;
    \item[(E)] there exists a directed path from $S$ to any $T \in {\bf T}$;
    \item[(F)] for every $V \in \mV \setminus (\{S\} \cup {\bf T})$ there exists a directed path from $S$ to $V$ and from $V$ to some terminal $T \in {\bf T}$.
\end{itemize}
The elements of $\mV$ are called \textbf{vertices} or \textbf{nodes}, and those of~$\mE$ are called \textbf{edges}. The elements of $\mV \setminus (\{S\} \cup {\bf T})$ are the \textbf{intermediate vertices}/\textbf{nodes}. The set of incoming and outgoing edges for a vertex $V$ are denoted by $\inn(V)$ and $\out(V)$, respectively. Their cardinalities are the  \textbf{indegree} and \textbf{outdegree} of $V$, which are denoted by $\ein(V)$ and $\eout(V)$, respectively.
\end{definition}

Our communication model is as follows:
all edges of a network $\mN$ can carry precisely one element from a set $\mA$ of cardinality at least 2, which we call the \textbf{alphabet}. The vertices of the network collect alphabet symbols over the incoming edges, process them according to \textit{functions}, and send the outputs over the outgoing edges. Vertices are \textit{memoryless} and transmissions are \textit{delay-free}.
We model errors as being introduced by an adversary $\adv$, who can corrupt the value of up to $t$ edges from a fixed set $\mU \subseteq \mE$. An alphabet symbol sent along one of the edges in $\mU$ can be changed to any other alphabet symbol at the discretion of the adversary. In particular, the noise we consider is \textit{not} probabilistic in nature, but rather worst-case: we focus on correcting \textit{any} error pattern that can be introduced by the adversary.
We call the pair $(\mN,\adv)$ an \textbf{adversarial network}.

It is well-known that an acyclic directed graph $(\mV,\mE)$ defines a partial order on the set of its edges, $\mE$. More precisely, $e_1 \in \mE$ \textbf{precedes} $e_2 \in \mE$ (in symbols, $e_1 \preccurlyeq e_2$) if there exists a directed path in $(\mV,\mE)$ whose first edge is~$e_1$ and whose last edge is $e_2$.
We may extend this partial order to a total order on $\mE$, which we fix once and for all and denote by $\le$. Important to note is that the results in this paper do not depend on the particular choice of $\le$.

\begin{definition}
Let $\mN=(\mV,\mE,S,{\bf T})$ be a network. A \textbf{network code} $\mF$ for $\mN$ is a family of functions $\{\mF_V \mid V \in \mV \setminus (\{S\} \cup {\bf T})\}$, where
$\mF_V: \mA^{\ein(V)} \to \mA^{\eout(V)}$ for all $V$.
\end{definition}

A network code $\mF$ describes how the vertices of a network~$\mN$ process the inputs received on the incoming edges. There is a unique interpretation for these operations thanks to the choice of the total order $\le$.

\begin{definition}\label{def:isolates}
Let $\mN=(\mV,\mE,S,{\bf T})$ be a network and let $\mU,\mU' \subseteq \mE$ be non-empty subsets. We say that $\mU$ \textbf{precedes}~$\mU'$ if every path from $S$ to an edge of $\mU'$ contains an edge from~$\mU$.
\end{definition}

Our next step is to define outer codes for a network and give necessary and sufficient conditions for decodability. We do this by introducing the notion of an adversarial channel as proposed in \cite[Section IV.B]{RK18}.

\begin{notation}
Let $(\mN,\adv)$ be an adversarial network with $\mN=(\mV,\mE,S,{\bf T})$ and let $\mU,\mU' \subseteq \mE$ be non-empty such that $\mU$ {precedes}~$\mU'$.
Let $\mF$ be a network code for~$\mN$. For $\mathbf{x} \in \mA^{|\mU|}$,
we denote by
\begin{equation} \label{nnn}\Omega[\mN,\adv,\mF,\mU \to \mU'](\mathbf{x}) \subseteq \mA^{|\mU'|}
\end{equation}
the set of vectors over the alphabet that can be exiting the edges of $\mU'$ when:
\begin{itemize}
    \item the coordinates of the vector $\mathbf{x}$ are the alphabet values entering the edges of $\mU$,
    \item vertices process information according to $\mF$,
    \item everything is interpreted according to the total order $\le$.
\end{itemize} Note that~\eqref{nnn} is well-defined because $\mU$ precedes $\mU'$. Furthermore, $\mU\cap \mU'$ need not be empty. We refer to the discussion following~\cite[Definition~41]{RK18}; see also~\cite[Example 42]{RK18}.
\end{notation}

\begin{example} \label{eee}
Let $(\mN,\adv)$ be the network in Figure~\ref{fig:eee}, where the edges are ordered according to their indices. We consider an adversary capable of corrupting up to one of the dashed edges. At each intermediate node $V\in \{ V_{1},V_{2}\}$, let $\mF_{V}$ be the identity function. Then, for example, for $\mathbf{x}=(x_1,x_2,x_3) \in \mA^3$ we have that
\[\Omega[\mN,\adv,\mF,\{e_1,e_2,e_3\} \to \{e_2,e_4,e_5\}](\mathbf{x}) \subseteq \mA^3\]
is the set of all alphabet vectors $\mathbf{y}=(y_1,y_2,y_3) \in \mA^3$ for which 
$\dH((y_2,y_1,y_3),(x_1,x_2,x_3)) \le 1$, where $\dH$ denotes the Hamming distance.
\begin{figure}[htbp]
\centering
\begin{tikzpicture}
\tikzset{vertex/.style = {shape=circle,draw,inner sep=0pt,minimum size=1.9em}}
\tikzset{nnode/.style = {shape=circle,fill=myg,draw,inner sep=0pt,minimum
size=1.9em}}
\tikzset{edge/.style = {->,> = stealth}}
\tikzset{dedge/.style = {densely dotted,->,> = stealth}}
\tikzset{ddedge/.style = {dashed,->,> = stealth}}

\node[vertex] (S1) {$S$};

\node[shape=coordinate,right=\mynodespace of S1] (K) {};

\node[nnode,above=0.5\mynodespace of K] (V1) {$V_1$};

\node[nnode,below=0.5\mynodespace of K] (V2) {$V_2$};

%\node[nnode,right=\mynodespace of S1] (V1) {$V_1$};
\node[vertex,right=\mynodespace of K] (T) {$T$};

\draw[ddedge,bend left=0] (S1)  to node[sloped,fill=white, inner sep=1pt]{\small $e_1$} (V1);

\draw[ddedge,bend left=0] (S1) to  node[sloped,fill=white, inner sep=1pt]{\small $e_2$} (T);

%\draw[ddedge,bend right=0] (S1)  to node[sloped,fill=white, inner sep=1pt]{\small $e_3$} (V2);

\draw[ddedge,bend right=0] (S1)  to node[sloped,fill=white, inner sep=1pt]{\small $e_3$} (V2);

\draw[edge,bend left=0] (V1)  to node[near start,sloped,fill=white, inner sep=1pt]{\small $e_4$} (T);

\draw[edge,bend left=0] (V2)  to node[sloped,fill=white, inner sep=1pt]{\small $e_{5}$} (T);

\end{tikzpicture} 
\caption{{{Network for Example~\ref{eee}.}}}\label{fig:eee}
\end{figure}
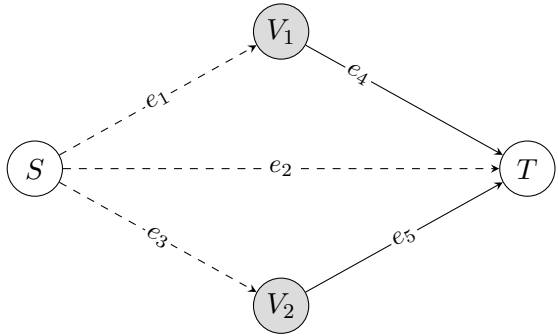
\end{example}

We now define error-correcting codes in the context of adversarial networks. Informally, codes are comprised of the alphabet vectors that may be emitted by the source.

\begin{definition}
An (\textbf{outer}) \textbf{code} for a network $\mN=(\mV,\mE,S,{\bf T})$ is a subset
$C \subseteq \mA^{\eout(S)}$ with $|C| \ge 1$.
If $\mF$ is a network code for $\mN$ and $\adv$ is an adversary, then we say that $C$ is \textbf{unambiguous} (or \textbf{good}) for $(\mN,\adv,\mF)$ if for all $\mathbf{x}, \mathbf{x}' \in C$ with $\mathbf{x} \neq \mathbf{x}'$ and for all $T \in {\bf T}$ we have
\begin{equation*}
    \Omega[\mN,\adv,\mF,\out(S) \to \inn(T)](\mathbf{x}) \, \cap   \Omega[\mN,\adv,\mF,\out(S) \to \inn(T)](\mathbf{x}') = \emptyset.
\end{equation*}
\end{definition}

The last condition in the above definition guarantees that every element of $C$ can be uniquely recovered by every terminal, despite the action of the adversary. Finally, we define the one-shot capacity of an adversarial network.

\begin{definition}
The (\textbf{one-shot}) \textbf{capacity}
of an adversarial network $(\mN,\adv)$ is the maximum $\alpha \in \R$ for which there exists
a network code $\mF$ and an unambiguous code $C$ for $(\mN,\adv,\mF)$ with $\alpha=\log_{|\mA|}(|C|)$.
We denote this maximum value by $\CA(\mN,\adv)$.
\end{definition}

In~\cite{RK18}, a general method was developed to ``lift'' bounds for Hamming-metric channels to the networking context. The method allows any classical coding bound to be lifted to the network setting. The next result states the lifted version of the well-known Singleton bound. Recall that an edge-cut between source $S$ and terminal $T$ is a set of edges whose removal would separate $S$ from~$T$.

\begin{theorem}[The Singleton Cut-Set Bound]
\label{cutset}
Let $\mN$ be a network with edge set $\mE$. Assume an adversary $\adv$ can corrupt up to $t \ge 0$ edges from a subset $\mU \subseteq \mE$. Then 
$$\CA(\mN,\adv) \le \min_{T \in {\bf T}} \min_{\mE'} \left( |\mE' \setminus \mU| +\max\{0, |\mE' \cap \mU|-2t\} \right),$$
where $\mE' \subseteq \mE$ ranges over all edge-cuts between $S$ and $T$.
\end{theorem}

\section{The Diamond Network: Achievability}
\label{sec:min-achiev}
% example
% scheme with reserved symbol

We present a minimal example of a network for which the best bound in~\cite{RK18}, namely the Singleton Cut-Set Bound, is not sharp. The example will serve to illustrate the necessity of performing \textit{partial} decoding at the intermediate nodes in order to achieve capacity.

\begin{example}[The Diamond Network] \label{minimal}
The network $\mD$ of Figure~\ref{figminimal} has one source $S$, one terminal~$T$, and two intermediate vertices $V_1$ and $V_2$.  The vertices are connected as in the figure. We consider an adversary $\adv_\mD$ able to corrupt at most one of the dashed edges, and we call the pair $(\mD,\adv_\mD)$ the \textbf{Diamond Network}.  
\begin{figure}[htbp]
\centering
\begin{tikzpicture}
\tikzset{vertex/.style = {shape=circle,draw,inner sep=0pt,minimum size=1.9em}}
\tikzset{nnode/.style = {shape=circle,fill=myg,draw,inner sep=0pt,minimum
size=1.9em}}
\tikzset{edge/.style = {->,> = stealth}}
\tikzset{dedge/.style = {densely dotted,->,> = stealth}}
\tikzset{ddedge/.style = {dashed,->,> = stealth}}

\node[vertex] (S1) {$S$};

\node[shape=coordinate,right=\mynodespace of S1] (K) {};

\node[nnode,above=0.5\mynodespace of K] (V1) {$V_1$};

\node[nnode,below=0.5\mynodespace of K] (V2) {$V_2$};

%\node[nnode,right=\mynodespace of S1] (V1) {$V_1$};
\node[vertex,right=\mynodespace of K] (T) {$T$};

\draw[ddedge,bend left=0] (S1)  to node[sloped,fill=white, inner sep=1pt]{\small $e_1$} (V1);

\draw[ddedge,bend left=15] (S1) to  node[sloped,fill=white, inner sep=1pt]{\small $e_2$} (V2);

%\draw[ddedge,bend right=0] (S1)  to node[sloped,fill=white, inner sep=1pt]{\small $e_3$} (V2);

\draw[ddedge,bend right=15] (S1)  to node[sloped,fill=white, inner sep=1pt]{\small $e_3$} (V2);

\draw[edge,bend left=0] (V1)  to node[near start,sloped,fill=white, inner sep=1pt]{\small $e_4$} (T);

\draw[edge,bend left=0] (V2)  to node[sloped,fill=white, inner sep=1pt]{\small $e_{5}$} (T);

\end{tikzpicture} 
\caption{{{The Diamond Network of Example~\ref{minimal}.}}}\label{figminimal}
\end{figure}
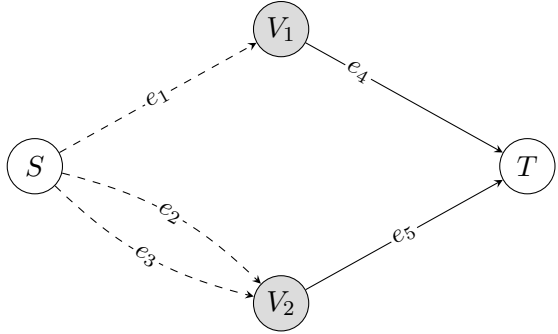
\end{example}

For the Diamond Network, the best bound among those proved in~\cite{RK18} is the Singleton Cut-Set Bound, given below.

\begin{corollary} \label{csD}
For the Diamond Network $(\mD,\adv_\mD)$, $\CA(\mD,\adv_\mD) \le 1$.
\end{corollary}

We will prove in this section and the next that the Diamond Network has capacity  
\begin{equation}\label{DNC}
    \CA(\mD,\adv_\mD) = \log_{|\mA|}(|\mA|-1).
\end{equation}
This shows that the bounds of~\cite{RK18} are not sharp. The results presented in the remainder of the paper offer an intuitive explanation for why Corollary~\ref{csD} is not sharp: in order to achieve the capacity of the Diamond Network, one alphabet symbol needs to be reserved to implement an adversary detection strategy. We will elaborate on this idea following the proof of achievability, given below.

\begin{proposition} \label{achiev}
For the Diamond Network $(\mD,\adv_\mD)$, $\CA(\mD,\adv_\mD) \ge \log_{|\mA|}(|\mA|-1)$.
\end{proposition}

\begin{proof}
We isolate a symbol $* \in \mA$ and define $\mA'=\mA \setminus \{*\}$. Consider the scheme where the source~$S$ can send any symbol of $\mA'$ via a three-times repetition code over its outgoing edges. Vertex $V_1$ simply forwards the received input, while vertex $V_2$ proceeds as follows: If the two received inputs coincide and are equal to $a \in \mA'$, then it forwards $a$. Otherwise, it transmits~$*$. It is not difficult to check that any symbol from $\mA'$ can be uniquely decoded, showing that the proposed scheme is unambiguous. This concludes the proof.
\end{proof}

The communication strategy on which the previous proof is based reserves an alphabet symbol $* \in \mA$ to pass information about the location of the adversary (more precisely,
the symbol~$*$ reveals whether or not the adversary is acting on the lower ``stream'' of the Diamond Network). Note that the source is not allowed to emit the reserved symbol~$*$, rendering $\log_{|\mA|}(|\mA|-1)$ the maximum rate achievable by this scheme.

 It is natural to then ask whether the reserved symbol~$*$ can simultaneously be a part of the source's codebook,
achieving a rate of $1=\log_{|\mA|}(|\mA|)$ transmitted message per single channel use. In the next section, we will formally answer this question in the negative; see Proposition~\ref{prop:conv}.
In Section~\ref{sec:mirrored-diamond}, we consider a modification of the Diamond Network and present a scheme where one symbol is reserved for adversary detection, but can nonetheless also be used as a message symbol.

\section{The Diamond Network: The Converse}
\label{sec:min-conv}
% notation - introduce formally, then given example within proof
% proof
In this section, we establish an inequality for the cardinality of any unambiguous code $C$ for the Diamond Network. The inequality is quadratic in the code's size and implies that $|C| \le |\mA|-1$. Together with Proposition~\ref{achiev}, this computes the exact capacity of $(\mD,\adv_\mD)$.

\begin{proposition} \label{prop:conv}
Let $\mF$ be a network code for $(\mD,\adv_\mD)$ 
and let $C \subseteq \mA^3$ be an outer code. If $C$ is unambiguous for $(\mD,\adv_\mD,\mF)$, then 
$$|C|^2 + |C| -1-|\mA|^2 \le 0.$$
In particular, we have $|C| \le |\mA|-1$.
\end{proposition}

\begin{proof} The argument is organized into various claims. We denote by $\pi: \mA^3 \to \mA$ the projection onto the first coordinate. 

\begin{claim}\label{clA}
We have $|\pi(C)|=|C|$.
\end{claim}
\begin{clproof}
This follows from the fact that $C \subseteq \mA^3$ must have minimum Hamming distance 3 in order to be unambiguous, as one can easily check.
\end{clproof}

\begin{claim}\label{clB}
The restriction of $\mF_{V_1}$ to $\pi(C)$ is injective.
\end{claim}
\begin{clproof}
Suppose by contradiction that there exist $\mathbf{x},\mathbf{y} \in C$ with $\pi(\mathbf{x}) \neq \pi(\mathbf{y})$ and $\mF_{V_1}(\pi(\mathbf{x})) = \mF_{V_1}(\pi(\mathbf{y}))$. Then it is easy to see that the sets
$\Omega[\mD,\adv_\mD,\mF,\out(S) \to \inn(T)](\mathbf{x})$
and 
$\Omega[\mD,\adv_\mD,\mF,\out(S) \to \inn(T)](\mathbf{y})$
 intersect non-trivially. Indeed, if $\mathbf{x}=(x_1,x_2,x_3)$ and $\mathbf{y}=(y_1,y_2,y_3)$, then the final output
$$(\mF_{V_1}(x_1), \mF_{V_2}(x_2,y_3)) \in \mA^2$$
belongs to both sets.
\end{clproof}

We now concentrate on the transfer from the edges in
$\{e_1,e_2,e_3\}$ to $e_5$. To simplify the notation, let
$$\Omega:=\Omega[\mD,\adv_\mD,\mF,\{e_1,e_2,e_3\} \to \{e_5\}],$$
which is well-defined because $\{e_1,e_2,e_3\}$ precedes $e_5$; see Definition~\ref{def:isolates}.

\begin{claim}\label{clC}
There exists at most one codeword $\mathbf{x} \in C$ for which the cardinality of 
$\Omega(\mathbf{x})$ is 1.
\end{claim}
\begin{clproof}
Towards a contradiction, suppose that there are $\mathbf{x},\mathbf{y} \in C$ with $\mathbf{x} \neq \mathbf{y}$ and $|\Omega(\mathbf{x})|=|\Omega(\mathbf{y})|=1$.
We write $\Omega':=\Omega[\mD,\adv_\mD,\mF,\{e_1,e_2,e_3\} \to \{e_2,e_3\}]$ and observe that  
$|\mF_{V_2}(\Omega'(\mathbf{x}))| = |\mF_{V_2}(\Omega'(\mathbf{y}))|=1$.
Let $\mathbf{x}=(x_1,x_2,x_3)$, $\mathbf{y}=(y_1,y_2,y_3)$.
Since $(x_2,x_3), (x_2,y_3) \in \Omega'(\mathbf{x})$ and
$(y_2,y_3), (x_2,y_3) \in \Omega'(\mathbf{y})$, we have
$$\mF_{V_2}(x_2,x_3)=\mF_{V_2}(x_2,y_3)=\mF_{V_2}(y_2,y_3).$$
By observing that the adversary may corrupt the symbol sent on $e_{1}$, this implies that the sets
$\Omega[\mD,\adv_\mD,\mF,\out(S) \to \inn(T)](\mathbf{x})$
and $\Omega[\mD,\adv_\mD,\mF,\out(S) \to \inn(T)](\mathbf{y})$
intersect non-trivially, a contradiction.
\end{clproof}

To simplify the notation further, denote the transfer from $S$ to $T$ by 
$$\Omega'':=\Omega[\mD,\adv_\mD,\mF,\{e_1,e_2,e_3\} \to \{e_4,e_5\}].$$ Since $C$ is unambiguous, we have
\begin{equation} \label{bb}
    \sum_{\mathbf{x} \in C} |\Omega''(\mathbf{x})| \le |\mA|^2.
\end{equation}
For all $\mathbf{x} \in C$, write
$\Omega''(\mathbf{x}) = \Omega_1''(\mathbf{x}) \cup \Omega_2''(\mathbf{x})$,
where
\begin{align*}
    \Omega_1''(\mathbf{x}) &=\{\mathbf{z} \in \Omega''(\mathbf{x}) \mid z_1=\mF_{V_1}(x_1)\}, \\
\Omega_2''(\mathbf{x})&=\{\mathbf{z} \in \Omega''(\mathbf{x}) \mid z_2=\mF_{V_2}(x_2,x_3)\}.
\end{align*}
By definition, we have 
$$|\Omega''(\mathbf{x})| = |\Omega''_1(\mathbf{x})| + |\Omega''_2(\mathbf{x})| -1.$$
Summing the previous identity over all $\mathbf{x} \in C$ and using Claims~\ref{clA}, \ref{clB} and \ref{clC} we find
\begin{align*}
    \sum_{\mathbf{x} \in C}|\Omega''(\mathbf{x})| &\ge 1+ 2(|C|-1) + \sum_{\mathbf{x} \in C} |C|  - |C|\\
    &= 2|C| -1 +|C|^2 -|C| \\
    &= |C|^2+|C|-1.
\end{align*}
Combining this with~\eqref{bb}, we find
$|C|^2+|C|-1 \le |\mA|^2$, which is the desired inequality.
\end{proof}

We can now compute the capacity of the Diamond Network by combining Propositions~\ref{achiev} and~\ref{prop:conv}.

\begin{theorem}
For the Diamond Network $(\mD,\adv_\mD)$, $\CA(\mD,\adv_\mD) = \log_{|\mA|}(|\mA|-1)$.
\end{theorem}

The Diamond Network is admittedly a small example. However, 
%there exist larger networks with one-shot capacity falling short of the Singleton Cut-Set Bound {\color{red} (NEED TO VERIFY)}, and 
we believe that it will provide valuable insight into the general behavior of the one-shot capacity of larger networks.

\section{The Mirrored Diamond Network}
\label{sec:mirrored-diamond}
It is interesting to observe that by adding a single edge to the Diamond Network as in Figure~\ref{fig:no-conv}, the capacity is exactly the one predicted by the Singleton Cut-Set Bound of Theorem~\ref{cutset}.
We call the network in Figure~\ref{fig:no-conv} the \textbf{Mirrored Diamond Network}. Again, the adversary can corrupt at most one edge from the dashed ones. The notation for the network-adversary pair is~$(\mS,\adv_\mS)$.

\begin{figure}[htbp]
\centering
\begin{tikzpicture}
\tikzset{vertex/.style = {shape=circle,draw,inner sep=0pt,minimum size=1.9em}}
\tikzset{nnode/.style = {shape=circle,fill=myg,draw,inner sep=0pt,minimum
size=1.9em}}
\tikzset{edge/.style = {->,> = stealth}}
\tikzset{dedge/.style = {densely dotted,->,> = stealth}}
\tikzset{ddedge/.style = {dashed,->,> = stealth}}

\node[vertex] (S1) {$S$};

\node[shape=coordinate,right=\mynodespace of S1] (K) {};

\node[nnode,above=0.5\mynodespace of K] (V1) {$V_1$};

\node[nnode,below=0.5\mynodespace of K] (V2) {$V_2$};

%\node[nnode,right=\mynodespace of S1] (V1) {$V_1$};
\node[vertex,right=\mynodespace of K] (T) {$T$};

\draw[ddedge,bend left=15] (S1)  to node[sloped,fill=white, inner sep=1pt]{\small $e_1$} (V1);

\draw[ddedge,bend right=15] (S1) to  node[sloped,fill=white, inner sep=1pt]{\small $e_2$} (V1);

\draw[ddedge,bend left=15] (S1)  to node[sloped,fill=white, inner sep=1pt]{\small $e_3$} (V2);

\draw[ddedge,bend right=15] (S1)  to node[sloped,fill=white, inner sep=1pt]{\small $e_4$} (V2);

\draw[edge,bend left=0] (V1)  to node[near start,sloped,fill=white, inner sep=1pt]{\small $e_5$} (T);

\draw[edge,bend left=0] (V2)  to node[sloped,fill=white, inner sep=1pt]{\small $e_{6}$} (T);

\end{tikzpicture} 
\caption{{{The Mirrored Diamond Network.}}}\label{fig:no-conv}
\end{figure}
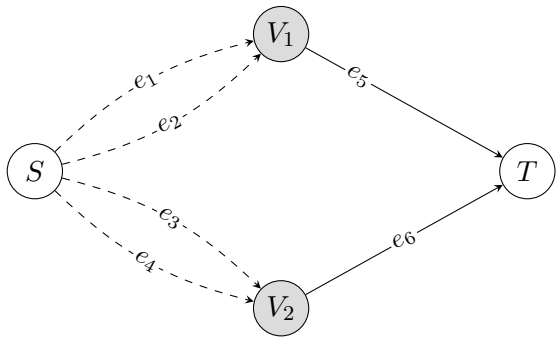

\begin{proposition}
\label{prop:sym-dia}
We have $\CA(\mS,\adv_\mS) = 1$.
\end{proposition}
\begin{proof}
By Theorem \ref{cutset}, $\CA(\mS,\adv_\mS) \leq 1$, so we need only prove achievability.
Select $* \in \mA$, and consider the scheme where the source $S$ sends any symbol of $\mA$ via a four-times repetition code. Vertices $V_1$ and $V_2$ both proceed as follows: If the two received inputs coincide and are equal to $a \in \mA$, the vertex forwards $a$; otherwise it transmits $*$. At $T$, if the received symbols match and are equal to $a\in \mA$, decode to $a$. Otherwise, decode to the symbol that is not equal to $*$. It is clear that any symbol from $\mA$ can be uniquely decoded, including $*$, showing that the proposed scheme is unambiguous.  This concludes the proof.
\end{proof}

Note that, as in the proof of Proposition~\ref{achiev}, the above scheme uses an alphabet symbol to pass information about the location of the adversary. In strong contrast with the Diamond Network however, in the Mirrored Diamond Network this strategy comes at no cost, as the ``reserved'' alphabet symbol can be used by the source like any other symbol.

\section{Two-Level Networks}
\label{sec:2-level}
% if there is no pathological node, then cut-set bound is achievable
% no pathological node is not necessary for achieving the cut-set bound (nice example!)

We initiate a systematic study of communication with restricted adversaries. Since a global treatment is out of reach at the moment, we start by concentrating on a small but sufficiently interesting family of highly structured networks. These are defined as follows.

\begin{definition}
A \textbf{two-level network} is a network $\mN=(\mV,\mE,S,\{T\})$ with a single terminal $T$ such that any path from $S$ to $T$ is of length $2$.
\end{definition}

\begin{figure}[htbp]
\centering
\begin{tikzpicture}
\tikzset{vertex/.style = {shape=circle,draw,inner sep=0pt,minimum size=1.9em}}
\tikzset{nnode/.style = {shape=circle,fill=myg,draw,inner sep=0pt,minimum
size=1.9em}}
\tikzset{edge/.style = {->,> = stealth}}
\tikzset{dedge/.style = {densely dotted,->,> = stealth}}
\tikzset{ddedge/.style = {dashed,->,> = stealth}}

\node[vertex] (S1) {$S$};

\node[shape=coordinate,right=\mynodespace of S1] (K) {};

\node[nnode,above=0.7\mynodespace of K] (V1) {$V_1$};
\node[nnode,above=0.1\mynodespace of K] (V2) {$V_2$};
\node[below=0\mynodespace of K]  {$\vdots$};
\node[nnode,below=0.6\mynodespace of K] (Vn) {$V_n$};

%\node[nnode,right=\mynodespace of S1] (V1) {$V_1$};
\node[vertex,right=\mynodespace of K] (T) {$T$};

\draw[ddedge,bend left=10] (S1)  to node[sloped,fill=white, inner sep=1pt]{} (V1);
\draw[ddedge,bend right=10] (S1) to  node[sloped,fill=white, inner sep=1pt]{} (V1);

\draw[ddedge,bend left=10] (S1)  to node[sloped,fill=white, inner sep=1pt]{} (V2);
\draw[ddedge,bend right=10] (S1)  to node[sloped,fill=white, inner sep=1pt]{} (V2);

\draw[ddedge,bend left=15] (S1)  to node[sloped,fill=white, inner sep=1pt]{} (Vn);
\draw[ddedge,bend left=5] (S1)  to node[sloped,fill=white, inner sep=1pt]{} (Vn);
\draw[ddedge,bend right=5] (S1)  to node[sloped,fill=white, inner sep=1pt]{} (Vn);
\draw[ddedge,bend right=16] (S1)  to node[sloped,fill=white, inner sep=1pt]{} (Vn);

\draw[edge,bend left=16] (V1)  to node[sloped, inner sep=1pt]{} (T);
\draw[edge,bend left=3] (V1)  to node[sloped, inner sep=1pt]{} (T);
\draw[edge,bend right=10] (V1)  to node[sloped, inner sep=1pt]{} (T);

\draw[edge,bend left=10] (V2)  to node[sloped, inner sep=1pt]{} (T);
\draw[edge,bend right=10] (V2)  to node[sloped, inner sep=1pt]{} (T);

\draw[edge,bend left=10] (Vn)  to node[sloped, inner sep=1pt]{} (T);
\draw[edge,bend right=10] (Vn)  to node[sloped, inner sep=1pt]{} (T);

\end{tikzpicture} 
\caption{\label{fig:2level-ex} An example of a two-level network where vulnerable edges are restricted to those in the first level. In general, there may be any number of edges between the source/sink and each intermediate node.}
\end{figure}
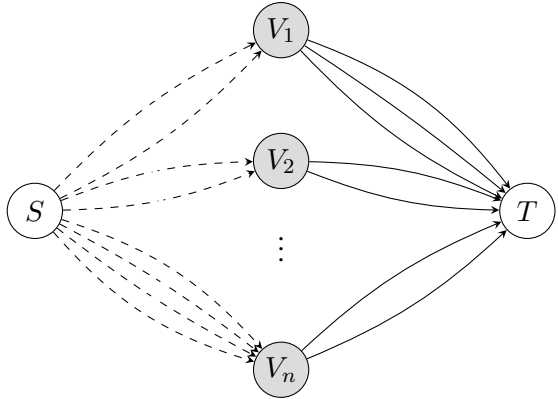

An example of a two-level network is given in Figure \ref{fig:2level-ex}. By applying the Singleton Cut-Set Bound of Theorem \ref{cutset} to two-level networks with vulnerable edges restricted to the first level, we establish the following bound.

\begin{theorem}
\label{cor:cut-set-2-level}
Consider a two-level network $\mN$ where the adversary $\adv$ can act on up to $t$ edges of the first level. Then, $C_{1}(\mN,\adv)$ is upper bounded by the following value:
\[\min_{\mV_{1},\mV_{2}}\left(\sum_{V_{i}\in \mV_{1}}\eout(V_{i})+\max\left\{0,\sum_{V_{i}\in \mV_{2}}\ein(V_{i})-2t\right\}\right),\] 
where the minimum is taken over all 2-partitions $\mV_{1}, \mV_{2}$ of the set of intermediate vertices $\{V_{1},\ldots,V_{n}\}$. 
\end{theorem}

To understand when the Singleton Cut-Set Bound is achievable in a two-level network, we introduce the following terminology.

\begin{definition}
Consider a network where an adversary can act simultaneously on up to $t$ edges. We call an intermediate vertex in the network \textbf{damming} if  
\[\eout(V_{i})+1 \leq \ein(V_{i})\leq \eout(V_{i})+2t-1.\]
\end{definition}

Notice that if the adversary can change at most one symbol, the above definition reduces to $\ein(V_{i})= \eout(V_{i})+1$; such a vertex is present in both the Diamond Network and the Mirrored Diamond Network.

\begin{theorem}
\label{thm:2-level}
In a two-level network where an adversary can act on up to $t$ edges of the first level, if no intermediate vertex is damming, then the Singleton Cut-Set Bound is achievable for sufficiently large alphabet size.
\end{theorem}

\begin{proof}
Suppose that no intermediate vertex is damming. That is, for every intermediate vertex $V_{i}\in \{V_{1},\ldots,V_{n}\}$, either $\ein(V_{i})\leq \eout(V_{i})$ or $\ein(V_{i})\geq \eout(V_{i})+2t$.
In this case, the bound of Theorem \ref{cor:cut-set-2-level} is achieved when 
\begin{align*}
    \mV_{1}&=\{V_{i} \mid \ein(V_{i})\geq\eout(V_{i})+2t\},\\
    \mV_{2}&=\{V_{i} \mid \ein(V_{i})\leq\eout(V_{i})\}.
\end{align*}
We exhibit a scheme that achieves the Singleton Cut-Set Bound. Choose a sufficiently large alphabet (determined by the required MDS codes below), and let $\mV_{1}$ and $\mV_{2}$ be as above. 

On $\eout(V_{i})+2t$ of the $\ein(V_{i})$ edges from the source, $S$, to vertex $V_{i}\in \mV_{1}$, send $\eout(V_{i})$ information symbols encoded using an MDS code of minimum distance $2t+1$; any extra edges from $S$ to $V_{i}$ may be disregarded. At vertex~$V_{i}$, decode the $\eout(V_{i})$ information symbols and forward them to the sink, $T$. 
Meanwhile, if $\sum_{V_{i}\in \mV_{2}}\ein(V_{i})> 2t$, encode $\sum_{V_{i}\in \mV_{2}}\ein(V_{i})-2t$ symbols using an MDS code with parameters $$\left[\sum_{V_{i}\in \mV_{2}}\ein(V_{i}),\sum_{V_{i}\in \mV_{2}}\ein(V_{i})-2t,2t+1\right],$$ and send this codeword along the edges from $S$ to the intermediate vertices in $\mV_{2}$. At the intermediate vertices, forward the received symbols; extra outgoing edges may be disregarded. If $\sum_{V_{i}\in \mV_{2}}\ein(V_{i})\leq 2t$, edges to $\mV_{2}$ may be disregarded.

At terminal $T$, decode the codeword sent through the vertices in $\mV_{2}$, if one exists, to retrieve $\sum_{V_{i}\in \mV_{2}}\ein(V_{i})-2t$ information symbols. An additional $\sum_{V_{i}\in \mV_{1}}\eout(V_{i})$ symbols were sent faithfully through the vertices of $\mV_{1}$. Altogether, this gives us $$\sum_{V_{i}\in \mV_{1}}\eout(V_{i})+\max\left\{0,\sum_{V_{i}\in \mV_{2}}\ein(V_{i})-2t\right\}$$ information symbols, achieving the Singleton Cut-Set Bound.   
\end{proof}

The results of Section~\ref{sec:mirrored-diamond}  demonstrate that the converse of Theorem \ref{thm:2-level} does not hold. Indeed, both intermediate vertices of the Mirrored Diamond Network are damming but its capacity is as predicted by the Singleton Cut-Set Bound; see Proposition~\ref{prop:sym-dia}.

\section{Discussion and Future Work}
\label{sec:conclusion}
We considered the problem of determining the one-shot capacity of communication networks with adversarial noise. In contrast with the typical scenario considered in the context of network coding, we allow the noise to affects only a subset of the network's edges.
We defined the Diamond Network and computed its capacity, illustrating that previously known cut-set bounds are not sharp in general. We then studied the family of two-level networks, giving a sufficient condition under which the Singleton Cut-Set Bound is sharp over a sufficiently large alphabet.

Natural problems inspired by these results are the complete characterization of two-level networks for which cut-set bounds are sharp, and development of techniques to derive upper bounds for the capacity of more general adversarial networks. These will be the subject of future work.

% Wishlist:
% \begin{itemize}
%     \item exact characterization of when the cut-set bound is achievable
%     \item capacity of networks for which cut-set is not achievable (more general strategy for achievability/converse)
% \end{itemize}

\bibliographystyle{IEEEtran}
\bibliography{ITW_2021}

\begin{thebibliography}{10}
\providecommand{\url}[1]{#1}
\csname url@samestyle\endcsname
\providecommand{\newblock}{\relax}
\providecommand{\bibinfo}[2]{#2}
\providecommand{\BIBentrySTDinterwordspacing}{\spaceskip=0pt\relax}
\providecommand{\BIBentryALTinterwordstretchfactor}{4}
\providecommand{\BIBentryALTinterwordspacing}{\spaceskip=\fontdimen2\font plus
\BIBentryALTinterwordstretchfactor\fontdimen3\font minus
  \fontdimen4\font\relax}
\providecommand{\BIBforeignlanguage}[2]{{%
\expandafter\ifx\csname l@#1\endcsname\relax
\typeout{** WARNING: IEEEtran.bst: No hyphenation pattern has been}%
\typeout{** loaded for the language `#1'. Using the pattern for}%
\typeout{** the default language instead.}%
\else
\language=\csname l@#1\endcsname
\fi
#2}}
\providecommand{\BIBdecl}{\relax}
\BIBdecl

\bibitem{YC06}
R.~W. Yeung and N.~Cai, ``Network error correction, {I}: Basic concepts and
  upper bounds,'' \emph{Communications in Inf. \& Systems}, vol.~6, no.~1, pp.
  19--35, 2006.

\bibitem{CY06}
N.~Cai and R.~W. Yeung, ``Network error correction, {II}: Lower bounds,''
  \emph{Communications in Inf. \& Systems}, vol.~6, no.~1, pp. 37--54, 2006.

\bibitem{YY07}
S.~Yang and R.~W. Yeung, ``Refined coding bounds for network error
  correction,'' in \emph{IEEE Inf. Theory Workshop on Inf. Theory for Wireless
  Networks}, 2007, pp. 1--5.

\bibitem{JLKHKM07}
S.~Jaggi, M.~Langberg, S.~Katti, T.~Ho, D.~Katabi, and M.~M{\'e}dard,
  ``Resilient network coding in the presence of byzantine adversaries,'' in
  \emph{26th IEEE Int'l Conference on Computer Communications}.\hskip 1em plus
  0.5em minus 0.4em\relax IEEE, 2007, pp. 616--624.

\bibitem{M07}
R.~Matsumoto, ``Construction algorithm for network error-correcting codes
  attaining the singleton bound,'' \emph{IEICE Trans. on Fundamentals of
  Electronics, Communications and Computer Sciences}, vol.~90, no.~9, pp.
  1729--1735, 2007.

\bibitem{YNY07}
S.~Yang, C.~K. Ngai, and R.~W. Yeung, ``Construction of linear network codes
  that achieve a refined {S}ingleton bound.'' in \emph{IEEE Int'l Symp. on Inf.
  Theory}, 2007, pp. 1576--1580.

\bibitem{YYZ08}
S.~Yang, R.~W. Yeung, and Z.~Zhang, ``Weight properties of network codes,''
  \emph{European Trans. on Telecommunications}, vol.~19, no.~4, pp. 371--383,
  2008.

\bibitem{RK18}
A.~Ravagnani and F.~R. Kschischang, ``Adversarial network coding,'' \emph{IEEE
  Trans. on Inf. Theory}, vol.~65, no.~1, pp. 198--219, 2018.

\bibitem{KK16}
O.~Kosut and J.~Kliewer, ``Network equivalence for a joint
  compound-arbitrarily-varying network model,'' in \emph{IEEE Inf. Theory
  Workshop}, 2016, pp. 141--145.

\bibitem{SBDP19}
N.~Sangwan, M.~Bakshi, B.~K. Dey, and V.~M. Prabhakaran, ``Multiple access
  channels with adversarial users,'' in \emph{IEEE Int'l Symp. on Inf.
  Theory}.\hskip 1em plus 0.5em minus 0.4em\relax IEEE, 2019, pp. 435--439.

\bibitem{BGKKY20}
A.~Beemer, E.~Graves, J.~Kliewer, O.~Kosut, and P.~Yu, ``Authentication and
  partial message correction over adversarial multiple-access channels,'' in
  \emph{IEEE Conference on Communications and Network Security}, 2020, pp.
  1--6.

\end{thebibliography}

\end{document}